\documentclass[10pt]{llncs}
\usepackage{fullpage}
\usepackage{graphicx}
\usepackage{amsfonts}
\usepackage{float}
\usepackage{amsmath}

\title{Tractability results for the Double-Cut-and-Join circular
  median problem}

\author{
  Ahmad Mahmoody-Ghaidary \inst{1,2}
  \and
  Cedric Chauve \inst{1}
  \and
  Ladislav Stacho \inst{1}
}

\institute{Department of Mathematics, Simon Fraser University, Burnaby
  (BC), Canada
\and Department of Computer Science, Brown University, Providence (RI), USA}

\begin{document}

\maketitle

%% ----------------------------------------------------------------

\begin{abstract}
  The circular median problem in the Dou\-ble-Cut-and-Join (DCJ)
  distance asks to find, for three given genomes, a fourth circular
  genome that minimizes the sum of the mutual distances with the three
  other ones. This problem has been shown to be NP-complete. We show
  here that, if the number of vertices of degree $3$ in the breakpoint
  graph of the three input genomes is fixed, then the problem is
  tractable\footnote{Version of \today. This paper is currently under
    peer-review. The results appeared in {\em Ahmad Mahmoody-Ghaidary,
      Tractability Results for the Double-Cut and Join
      Multichromosomal Problem, MSc thesis, Department of Mathematics,
      Simon Fraser University, 2011}.}.
\end{abstract}

%%%%%%%%%%%%%%%%%%%%%%%%%%%%%%%%%%%%%%%%%%%%%%%%%%%%%%%%%%
\section{Introduction} \label{sec:intro}

Comparative genomics has been an important source of combinatorial and
algorithmic questions during the last 20 years, especially the
computation of genomic distances and ancestral genomes, as illustrated
by the recent book of Fertin {\em et al.}  \cite{Fertin2009}. Among
these problems, the {\em median} problem is of particular interest:
while the distance problem is tractable in many models, the median
problem is its simplest natural extension (a distance is a function of
two genomes, while the median score is a function of three genomes)
and is computationally intractable in most models. Computing median is
at the heart of inferring gene order phylogenies and ancestral gene
orders~\cite{Murphy2005,Lin2010,Xu2011}. This motivated research on
tractability issues of genomic median problems, well summarized in the
recent paper~\cite{Tannier2009}, as well as on practical algorithms to
address it (see~\cite{Xu2008,Zhang2009,Xu2009a,Xu2009b} and references
there).

Roughly speaking, the median problem is as follows: given three
genomes $G_1$, $G_2$ and $G_3$ and a genomic distance model $d$, find
a genome $M$ that minimizes the \emph{cost} of $d$ over $G_1, G_2,
G_3$ defined by $d(M,G_1)+d(M,G_2)+d(M,G_3)$. It is in fact an
ancestral genome reconstruction problem, as $M$ can be seen as the
last-common ancestor of $G_1$ and $G_2$, with $G_3$ acting as {\em
  outgroup} (i.e. a genome whose last common ancestor with $G_1$ and
$G_2$ is an ancestor of $M$). In~\cite{Tannier2009}, Tannier {\em et
  al.} explored several variants, based on different models of genomes
(linear, circular or mixed, see Section~\ref{sec:prelim}) and of
genomic distances (Breakpoint, Double-Cut-and-Join, Reversals,
\dots). In particular, they showed that if $d$ is the
Double-Cut-and-Join (DCJ) distance, which is currently the most widely
used genomic distance, then computing a circular or mixed median is
NP-complete. In fact, the only known tractable median problem is the
mixed breakpoint median: $d$ is the breakpoint distance and the median
can contain both linear and circular chromosomes. From a combinatorial
point of view, the central object in the DCJ model is the {\em
  breakpoint graph}: the DCJ distance between two genomes with $n$
genes is indeed easily obtained from the number of cycles and paths
containing odd number of vertices (odd paths) in this
graph~\cite{Yancopoulos2005,Bergeron2006}. Recent progress in
understanding properties of this graph, and especially of the family
of {\em adequat subgraphs}, lead Xu to introduce algorithms to compute
DCJ median genomes which are efficient on real data, but do not define
well characterized classes of tractable
instances~\cite{Xu2008,Xu2009a,Xu2009b,Xu2011}.

In the present work, we show the following result: if the breakpoint
graph of three genomes contains a constant number of vertices of
degree 3, then computing a DCJ circular median is tractable. To the
best of our knowledge, this is the first result defining an explicit
non-trivial class of tractable instances related to the DCJ median
problem. In Section~\ref{sec:prelim}, we define precisely
combinatorial representations of genomes, the DCJ distance, breakpoint
graphs and the problem we addressed here. In
Section~\ref{sec:results}, we state and prove our main result.
%% Moreover, aside of its theoretical interest, this result has
%% practical implications, as we illustrate in
%% Section~\ref{sec:algorithm}, where we design a heuristic that relies
%% removing a subset of edges to decrease their number and then applying
%% our algorithm, and use this algorithm to study a vertebrate genome
%% dataset.

%% ; this then gives a lower bound to the cost of a circular
%% median, which can be used in a branch-and-bound algorithm, or even in
%% the framework developed by Xu~\cite{Xu2008,Xu2009a,Xu2009b}.

%%%%%%%%%%%%%%%%%%%%%%%%%%%%%%%%%%%%%%%%%%%%%%%%%%%%%%%%%%
\section{Preliminaries} \label{sec:prelim}

\paragraph{Genes, genomes and breakpoint graph}
Let $A = \{1, 2, \ldots, n\}$ represent a set of $n$ {\em
  genes}\footnote{The term {\em gene} is used here in a generic way,
  and might include other genomic markers such as synteny/orthology
  blocks for example.}. Each gene $i$ has a {\em head} $i_h$ and a
{\em tail} $i_t$. From now, we assume $A$ always contains $n$ genes.

A {\em genome} $G$, with gene set $A$, is encoded by the order and
orientation of its genes along its chromosomes (i.e. its {\em gene
  order}), or equivalently by the set of the adjacencies between its
gene extremities, that can naturally be represented by a {\em
  matching} on the set of vertices $V(G) = \{i_h, i_t | 1\leq i \leq n
\}$ (Fig.~\ref{genomebreak} (a)). The connected components of the
graph whose vertices are $V(G)$ and edges are the disjoint union of
the edges of $G$ and the edges $\{i_t, i_h\}$ (forcing gene
extremities for a given gene to be contiguous) form the {\em
  chromosomes} of $G$ (Fig.~\ref{genomebreak} (b)). A chromosome is
\emph{linear} if it is a path and \emph{circular} if it is a
cycle. $G$ is {\em circular} if it contains only circular chromosomes
(perfect matching), {\em linear} if it contains only linear
chromosomes, and {\em mixed} otherwise.  Fig.~\ref{genomebreak}(a,b)
illustrates this view of genomes as matchings.

The {\em breakpoint graph} $B(G_1, \ldots, G_m)$ of $m$ genomes $G_1,
\ldots, G_m$ on $A$ is the disjoint union of these genomes, i.e. the
graph with vertex set $V(A)$ and edges given by the matchings defining
these $m$ genomes. Following the usual convention, we consider that
edges in this graph are colored, with color $c_i$ assigned to genome
$i$ ($1 \leq i \leq m$); it results that $B(G_1, \ldots, G_m)$ can
have multiple edges of different colors (see Fig.~\ref{genomebreak}
(c)). 

\begin{figure}
\begin{center}
  \includegraphics[scale=0.45]{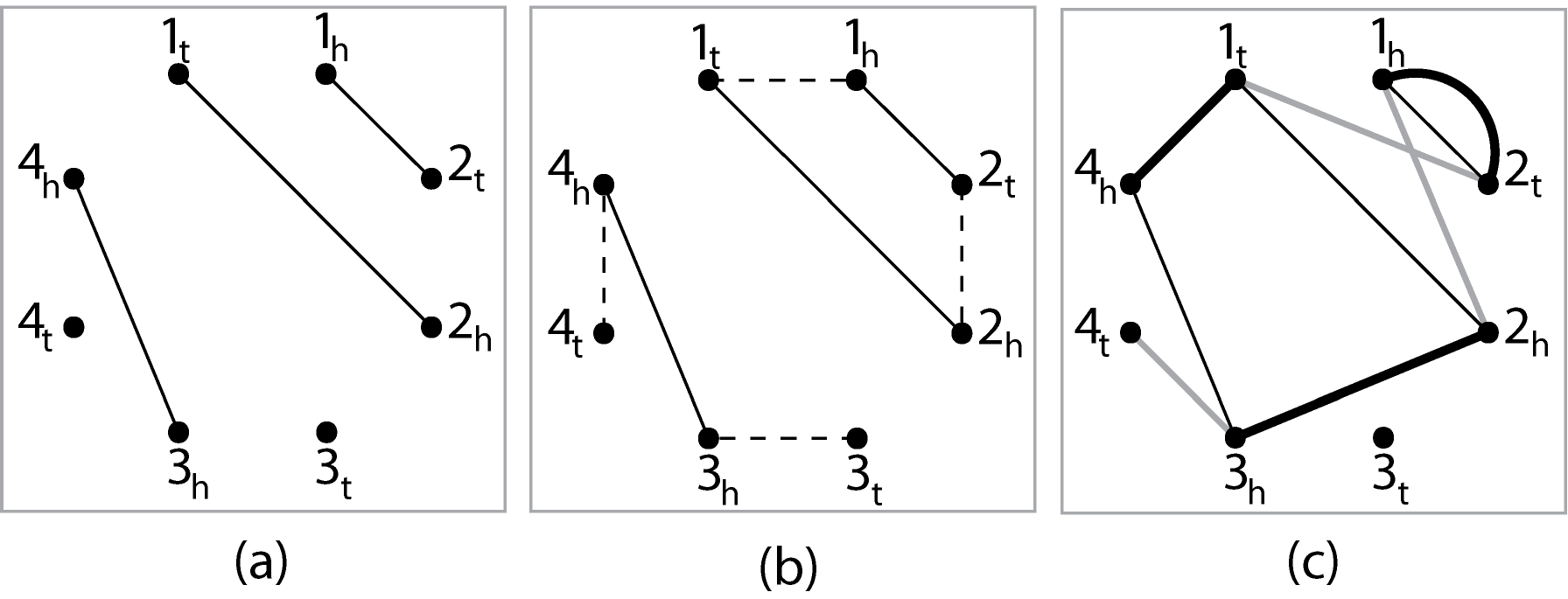}
 \vspace*{-3mm}\caption{(a) A genome on 4 genes, with two chromosomes,
   one circular chromosome with gene order $(1\ 2)$ and one linear
   chromosome with gene order $(4\ -3)$, where the sign $-$ indicates
   a reverse orientation. (b) The same genomes with added dashed edges
   connecting gene extremities: every connected component of the
   resulting graph is a chromosome. (c) The breakpoint graph of three
   genomes (whose edges are respectively light gray, thin black and
   thick black) on 4 genes.}
    \label{genomebreak}
\end{center}

\end{figure}

\paragraph{DCJ distance and median}
Given two genomes $G$ and $M$ on $A$, with $M$ being a circular
genome, the {\em DCJ distance} $d_{\text{DCJ}}(G,M)$ is given by
\begin{equation}\label{distancefunc}
d_{DCJ}(G, M) = n - c(G, M),
\end{equation}
where $c(G,M)$ is the number of cycles in $B(G,M)$. So larger $c(G,
M)$ implies smaller distance $d_{DCJ}(G, M)$. The general definition
of the DCJ distance (when both $G$ and $M$ are mixed genomes) also
requires to consider the odd paths\footnote{ An \emph{odd}
  (resp.\emph{even}) subgraph is a subgraph with an odd (resp. even)
  number of vertices.} of the breakpoint graph~\cite{Bergeron2006},
but it is easy to see that the breakpoint graph does not contain odd
path if at least one genome is circular. Note that the edges on a
cycle in $B(G, M)$ are alternatively from $M$ and $G$. An \emph{$(G,
  M)$-alternating cycle} is an even cycle with edges in $M$ and $G$
alternatively. For simplicity, we may sometimes only call such cycles
\emph{alternating cycles}.

A {\em DCJ circular median} for three genomes $G_1$, $G_2$, and $G_3$,
or alternatively for their breakpoint graph $B(G_1,G_2,G_3)$, is a
circular genome $M$ which minimizes 
$$\sum\limits_{i = 1}^{3}d_{DCJ}(G_i, M)=3n-\sum\limits_{i =
  1}^{3}c(G_i, M)$$. So a circular genome $M$ which maximizes the the
total number of $(M, G_i)$-alternating cycles (for an $i \in \{1, 2,
3\}$ ) is a DCJ circular median.

\paragraph{Terminology}
From now, by \emph{median} we always mean \emph{DCJ circular median}.
We denote also by $m(B)$ the sum $d_{DCJ}(G_1, M)+d_{DCJ}(G_2,
M)+d_{DCJ}(G_3, M)$ for a median $M$.

Let $B = B(G_1, G_2, G_3)$ be a breakpoint graph, and let $M$ be a
median of $B$.  The graph $B_M(G_1, G_2, G_3)=B \cup M$ (also denoted
by $B_M$ when the context is clear) is called the \emph{median graph}
of $B$ with the DCJ circular median genome $M$ (using disjoint union).
The edges in $G_1 \cup G_2 \cup G_3$ are called {\em colored
  edges}\index{colored edge}, and edges in $M$ are called {\em median
  edges}\index{median!edge}. 

A $k$-cycle in $B_M$ is an $(M,G_i)$-alternating cycle of length $k$,
for some $i$, in $B_M$. We denote the total number of alternating
cycles for a median graph $M$ of a breakpoint graph $B$ by
$\text{cyc}(B)$\footnote{Note that $\text{cyc}(B)$ does not depend
  only on the topology of $B$, but also on the colors of its
  edges. Moreover, for different medians $M$ of $B$, $B_M$ has the
  same number of alternating cycles, so $\text{cyc}(B)$ does not
  depend of a particular median.}.  If $H$ is a subgraph of $B$, then
$\text{cyc}(H)$ is the maximum number of alternating cycles composed
of edges in $H$, taken over all matchings in $H$.

A terminal vertex in a graph is a vertex of degree $1$. A subgraph of
$B$ is said to be {\em isomorphic to $C_k$ (resp. $P_k$)} if it is a
cycle (resp. path) on $k$ vertices.

\begin{remark}
The problem we consider in the present work is to compute a DCJ
circular median of three given genomes, or equivalently to find a
matching in $B$ that maximizes the number of alternating cycles. From
this point of view this is a purely graph theoretical problem that can
be extended naturally to any edge-colored graph, with the convention
that if the graph has an odd number of vertices, then exactly one
vertex does not belong to the matching.
\end{remark}

\paragraph{Shrinking in a breakpoint graph}
{\em Shrinking} a pair of vertices $\{u, v\}$ or an edge with end
vertices $u$ and $v$ was defined in~\cite{Xu2008}. It consists of
three steps: {\em (1)} removing all edges between $u$ and $v$ (if
there is any), {\em (2)} identifying the remaining edges incident to
both $u$ and $v$ and with same color, {\em (3)} removing $u$ and
$v$. We denote the resulting graph by $B\cdot\{u,v\}$
(Fig.~\ref{shrink}).

%\vspace{-4cm}
\begin{figure}{
    \begin{center}
      \includegraphics[scale=0.5]{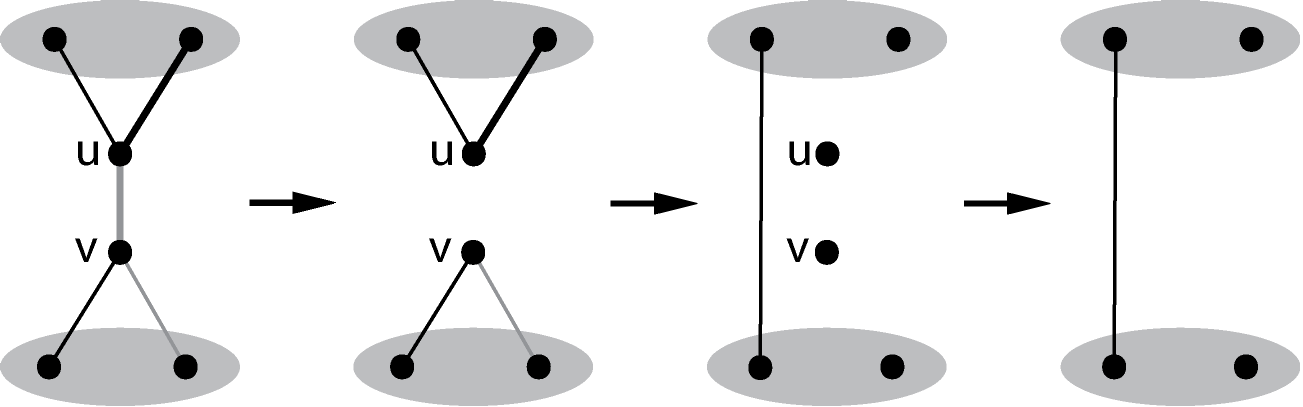}
      \caption{Illustration of the shrinking of a pair $\{u, v\}$ of
        vertices of a breakpoint graph.}
      \label{shrink}
    \end{center}
}\end{figure}

%{\noindent\bf Ahmad: was it a result stated in some paper of Xu. If
 % yes, we should cite it. If no, we should also state it is new, which
 % justifies to give a proof.}

\begin{proposition} \label{Shrink}
  Let $B$ be the breakpoint graph of genomes $G_1, \ldots, G_m$, and
  $u, v \in V(B)$. Suppose that there are $k$ colored edges between
  $u$ and $v$. If there exists a median $M$ containing the edge $uv$,
  then $\text{cyc}(B) = \text{cyc}(B\cdot\{u, v\})+k$.
\end{proposition}

\begin{proof} 
 Consider a median $M$ which contains the edge $uv$ (which implies
 that both $u$ and $v$ are in the same alternating cycle in
 $B_M$). Let $B' = B\cdot\{u, v\}$, $M' = M - \{u, v\}$ (the graph
 obtained from $M$ by removing $u$, $v$ and the edge $uv$).

 Let $C$ be an alternating cycle in $B_M$. If $C$ does not contain $u$
 and $v$, then, obviously, $C$ does not contain any of the $k$ edges
 between $u$ and $v$. Thus, $C$ remains unchanged in $B'_{M'}$. Assume
 now that $C$ contains $uv$. If the length of $C$ is larger than 2,
 shrinking $\{u, v\}$ results in a cycle with smaller length in
 $B'_{M'}$ (the length decreases by $2$). Otherwise, if $C$ has length
 2, it disappears in $B'_{M'}$. Thus the number of alternating cycles
 which disappear in $B' \cup M'$ is $k$, since there are $k$ edges
 between $u$ and $v$. Therefore, $\text{cyc}(B) \leq
 \text{cyc}(B\cdot\{u, v\})+k$.
  
  Now suppose $N'$ is a median of $B'$.  By a similar argument, if $N
  = N' \cup \{u, v\}$, then $B_N$ has $\text{cyc}(B\cdot\{u, v\})+k$
  alternating cycles. So, $\text{cyc}(B) \geq \text{cyc}(B\cdot\{u,
  v\})+k$, and, as $\text{cyc}(B) = \text{cyc}(B\cdot\{u, v\})+k$, we
  have that $M'$ (resp. $N$) is a median of $B'$ (resp. $B$).  %% \qed
\end{proof}

%%%%%%%%%%%%%%%%%%%%%%%%%%%%%%%%%%%%%%%%%%%%%%%%%%%%%%%%%%
\section{A class of tractable instances}\label{sec:results}

Our main theoretical result is the definition of a large class of
tractable instances for the median problem, namely the ones whose
breakpoint graph contains few vertices of degree $3$. Obviously, the
median problem for three genomes involves a breakpoint graph with
maximum degree $3$. We show here that the hardness of the problem is due
to these vertices of degree $3$.

\begin{theorem}\label{thm:main}
  Let $G_1,G_2$, and $G_3$ be three genomes. If there exists a median
  of $B=B(G_1,G_2,G_3)$ with at most $\ell$ edges whose both
  end-vertices are of degree $3$ in $B$, then computing such a median
  can be done in time $O(n^{3} \cdot (\ell+1) \cdot(3^{m}\cdot
  m^{2\ell}+1))$, where $m$ is the number of vertices of degree $3$, and
  $n$ is the number of genes in $G_1,\ G_2,\ G_3$.
\end{theorem}

\begin{remark}\label{rem:corollaries}
  Note that, as corollaries of this theorem, we have in particular that, 
  \begin{enumerate}
  \item if $m$ is bounded, then computing a median is tractable,
  \item if $\ell$ is bounded, then computing a median is Fixed-Parameter
    Tractable (FPT) (see~\cite{niedermeier-2008} for a reference on FPT
    algorithms) with parameter $m$.  
  \end{enumerate}
  Moreover, if $m$ is not bounded, we can remove some edges incident
  to vertices of degree $3$, so that in the new instance the number of
  vertices of degree $3$ is bounded. Now, by point 1 above, there is a
  polynomial time algorithm which computes the median of the new
  instance. %% As we show in Section~\ref{sec:algorithm} this gives a
  %% useful way to compute a lower bound on the cost of a median.
\end{remark} 

Informally, to prove Theorem~\ref{thm:main}, we first consider the
case where $B$ is a collection of cycles and paths (i.e. has maximum
degree $2$) and show that a median can be computed in polynomial
time. Next, we consider all possibilities (configurations) for
matching vertices of degree $3$ as median edges. For each configuration,
we reduce the breakpoint graph by shrinking and removing some edges to
obtain a graph whose connected components are paths or cycles. Having
computed all possible configurations for vertices of degree $3$ and
being able to compute a median for all resulting graphs lead to
Theorem~\ref{thm:main}.

%% CEDRIC: this figures arrives too early to really help the reader. I moved it later
%% \begin{figure}[H]{
%% \begin{center}
%% \includegraphics[scale=0.25]{comps.pdf}
%% \caption{{\bf (e)} Each even component has no crossing edge, {\bf (o)} For each odd component $H_1$, there is exactly another odd component $H_2$, such that $(H_1 \cup H_2)$ has no crossing edge.}
%% \label{comps}
%% \end{center}
%% }\end{figure}

From now, $G_1$, $G_2$, and $G_3$ are mixed genomes on $n$ genes, and
$M$ is a median of these genomes, unless otherwise specified. We
denote their breakpoint graph by $B$, and the median graph by $B_M$.

%%%%%%%%%%%%%%%%%%%%%%%%%%%%%%%%%%%%%%%%%%%%%%%%%%%%%%%%%%%%%%%%%%%%%
\subsection{Preliminary results}

We first introduce two useful lemmas that give lower bounds on the
function $\text{cyc}$ in various cases.

\begin{lemma}\label{subpc}
  If $B$ is isomorphic to $P_k$ or $C_{2k}$, for $k\geq 1$, then for
  every subgraph $H \subseteq B$, $\text{cyc}(H) \geq
  \frac{|E(H)|}{2}$.
\end{lemma}

\begin{proof}
  Consider the path $P_k = u_1u_2\ldots u_k$. Let $M$ be the matching
  consisting of the edges $u_1u_2, u_3u_4,\ldots$, and $u_{t-1}u_t$,
  where $t = 2\lfloor \frac{k}{2} \rfloor$. Obviously, the number of
  alternating cycles in $P_k \cup M$ is $\lfloor k/2 \rfloor$, so
  $\text{cyc}(P_k) \geq \frac{k}{2} \geq \frac{|E(P_k)|}{2}$.
  Similarly $\text{cyc}(C_{2k}) \geq k = \frac{|E(C_{2k})|}{2}$.  See
  Fig.~\ref{pc}.

  Any proper subgraph $H \subset P_k$ or $C_{2k}$ is a union of
  disjoint paths. If we take the union of matchings described above
  for each of these paths and call it $M$, there are at least
  $|E(H)|/2$ alternating cycles in $H \cup M$. Therefore for any
  subgraph $H \subseteq H$, $\text{cyc}(H) \geq
  \frac{|E(H)|}{2}$. %% \qed
\end{proof}

\begin{figure}[h]{
    \begin{center}
      \includegraphics[scale=0.75]{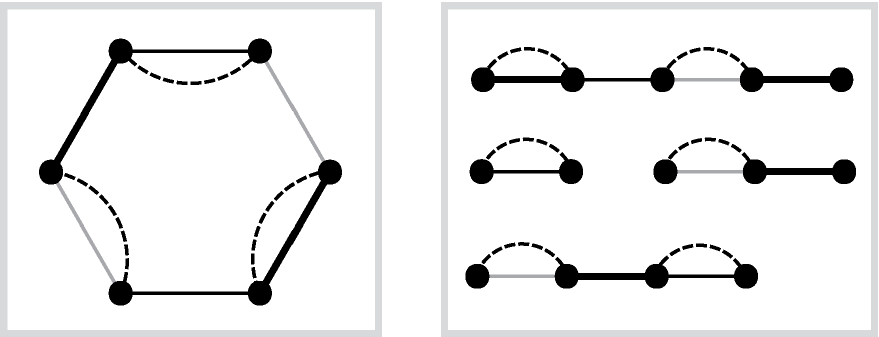}
      \caption{Median edges (dashed) for cycles and union of disjoint paths.}
      \label{pc}
    \end{center}
}\end{figure}

\begin{definition}
  Let $S$ and $T$ be two subgraphs of $B$. $T$ is an {\em
    alternating-subdivision} of $S$ if we can obtain an isomorphic
  copy of $T$ from $S$ as follows: subdivide each edge $e = \{a, b\}$
  by an even (possibly zero) number of vertices resulting in a path
  $av_1v_2\ldots v_{2k}b$, then remove every second edge, i.e.,
  $v_1v_2, v_3v_4, \ldots, v_{2k-1}v_{2k}$. We call the removed edges
  a {\em completing matching} for $T$ respective to $S$.
\end{definition}

In the previous definition, note that there might be more than one way
to obtain an isomorphic copy of $T$ from $S$, and consequently,
completing matching is not necessarily unique.
%So if $T$ is an alternating-subdivision of $S$, there exists a,
%possibly non-perfect, matching of $T$ such that (1) adding the edges
%of this matching to $T$ and (2) replacing by a single edge any maximal
%alternating path starting and ending by edges of $T$, results in
%$S$. Note that such a matching might not be unique. We call such a
%matching a {\em completing matching for $T$ respective to $S$}.

%% CEDRIC
%% {\noindent\bf There was a slight problem here: given $S$ and $T$,
%%   $\text{Rem}(S,T)$ might not be unique. Moreover, the removed edges
%%   do not exist in $S$ and $T$, so it is weird to give them a name.}

\begin{lemma}\label{alsub}
  If $T$ is an alternating-subdivision of $S$, then $\text{cyc}(T)
  \geq \text{cyc}(S)$.
\end{lemma}

\begin{proof}
  Let $M$ be a median of $S$, $M'$ an arbitrary completing matching
  for $T$ respective to $S$ and $M'' = M \cup M'$. $M''$ is a perfect
  matching of $T$, and each alternating cycle in $S \cup M$ defines a
  unique alternating cycle in $T \cup M''$ which implies that
  $\text{cyc}(T) \geq \text{cyc}(S)$ (see Fig.~\ref{proofalsub}).
  %% \qed
\end{proof} 

\begin{figure}[h]
  \begin{center}(a)
    \includegraphics[scale=0.55]{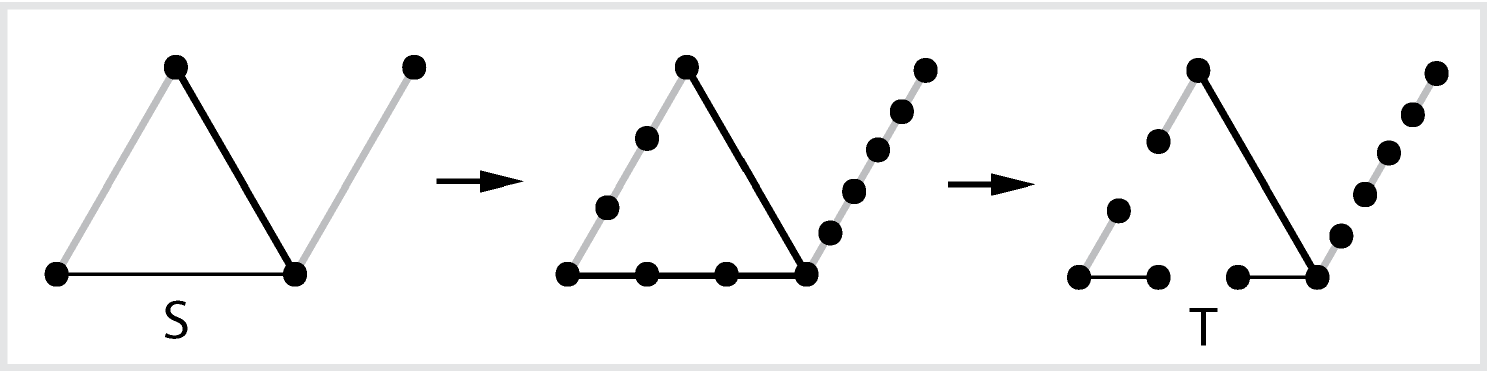} \\ (b)
    \includegraphics[scale=0.55]{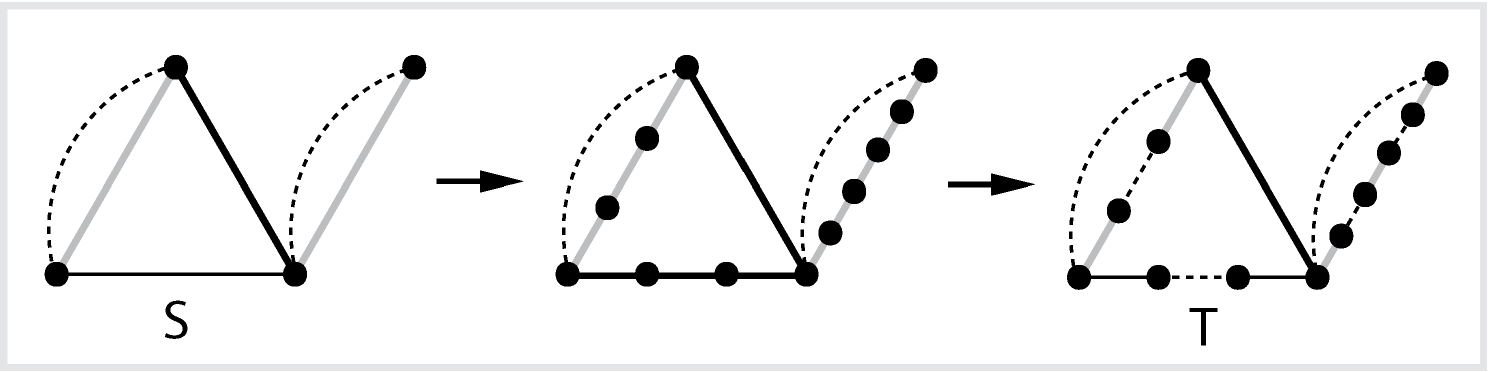}
    \caption{(a) Obtaining $T$ as an alternating-subdivision of
      $S$. (b) Obtaining a matching of $T$ from a median of $S$ (the
      dashed edges are the median edges and the edges of a completing
      matching for $T$ respective to $S$).}
     \label{sub}
     \label{proofalsub}
  \end{center}
 \end{figure}

%% \begin{figure}[h]
%%  \begin{center}
%%    \includegraphics[scale=0.75]{asm.pdf}
%%    \caption{Obtaining a matching of $T$ from a median of $S$ (dashed
%%      edges are the median edges)}

%%  \end{center}
%% \end{figure}

%%%%%%%%%%%%%%%%%%%%%%%%%%%%%%%%%%%%%%%%%%%%%%%%%%%%%%%%%%%%%%%%%%%%%
\subsection{Independence of arbitrary paths and even cycles}

In this section we introduce the fundamental notion of {\em
  independence} of connected components of cycles and paths in a
breakpoint graph. 

\begin{definition}\label{def:crossing-indep}
  Let $H$ be a subgraph of $B$. An {\em $H$-crossing
    edge}\index{crossing edge} in a median graph $B_M$ is a median
  edge which connects a vertex in $V(H)$ to a vertex in $V(B) -
  V(H)$. An {\em $H$-crossing cycle}\index{crossing cycle} is an
  alternating cycle which contains at least one $H$-crossing edge.
  The subgraph $H$ is {\em $k$-independent} if there is a median $M$
  for $B$ such that the number of $H$-crossing edges in $B_M$ is at
  most $k$.
\end{definition}

%%%%%%%%
\begin{proposition}\label{indep}
  Let $H$ be a connected component of $B$. If $H$ is isomorphic to
  $P_{2k}$ or $C_{2k}$, for $k\geq 1$, then $H$ is 0-independent.
\end{proposition}

\begin{proof}
  Let $M$ be a median of $B$. Suppose $M$ has $\ell$ $H$-crossing
  edges in $B_M$. If $\ell = 0$, then we are done, so assume that
  $\ell > 0$. Since $H$ has an even number of vertices, $\ell$ is even
  and $\ell \geq 2$. Because $H$ is a connected component in $B$, each
  $H$-crossing cycle contains an even number of $H$-crossing edges.
  
  Let $C_{M,H}$ be the set of all $H$-crossing cycles in $B_M$, and
  $E^{cr}_{M,H}$ be the set of all $H$-crossing edges in $B_M$.  Let
  $\text{X}(M)$ be the set of colored edges in all cycles of
  $C_{M,H}$, and $\text{Y}(M)$ be the set of all $H$-crossing edges in
  all cycles of $C_{M,H}$.

  {\smallskip\noindent\em Case 1.} If there is no $H$-crossing cycle,
  i.e., $C_{M,H} = X(M)=Y(M) = \emptyset$, we modify $M$ by
  removing all $H$-crossing edges, and re-matching the vertices inside
  of $H$ together and outside of $H$ together. Since $\ell$ is even,
  this is always possible and we get a median with no $H$-crossing
  edge.  

  {\smallskip\noindent\em Case 2.} From now, we assume that there
  exists at least one $H$-crossing cycle. The remainder of the proof
  relies on a transformation on $M$ that reduces the number of edges
  in $H$-crossing cycles, leading to a median with no $H$-crossing
  edge.

  {\smallskip\noindent\em Step 1.} The first step consists of
  choosing, for each $H$-crossing cycle, an arbitrary colored edge in
  $H$ incident to an $H$-crossing edge from this cycle. Let $S$ be the
  subgraph of $B$ induced by these chosen colored edges and
  $T=\text{X}(M) - S$.

  %% See Fig.~\ref{result}.
  %% %% CEDRIC. I do not think this figure ads a lot of clarity.
  %% \begin{figure}
  %%   \begin{center}
  %%     \includegraphics[scale=0.75]{result.pdf}
  %%     \caption{Dashed lines are the median edges (the bold dashed edge
  %%       is in $M'$ as in the proof of Proposition~\ref{indep}). The
  %%       edges of $S$ and $T$ are shown by solid and bold solid edges,
  %%       respectively (in dark gray areas). Note that $T$ can have
  %%       edges in $H$ and $B-H$.}
  %%       \label{result}
  %%   \end{center}
  %% \end{figure}

  {\smallskip\noindent\em Claim 1. $T$ is an alternating-subdivision
    of $S$.} For a vertex $x \in V(S)$ let $x_M$ be the neighbor of
  $x$ in $M$. If $u, v \in V(S)$ and $uv \in E(S)$ then, by
  definition, $uv$ is a colored edge of an $H$-crossing cycle which is
  incident to an $H$-crossing edge. Therefore, there is an alternating
  path from $u_M$ to $v_M$, with alternating colored and median edges
  from that cycle. If this path has $t$ colored edges, we subdivide
  the edge $uv$ using $2t-2$ vertices and remove every second
  edge. Proceeding in this way for every edge $uv \in E(S)$ we obtain
  an alternating-subdivision $T$ of $S$.

  {\smallskip\noindent\em Claim 2. $\text{cyc}(T) \geq
    |C_{M,H}|/2$.} First, as every colored edge is in at most one
  alternating cycle and two edges of the same color are not incident
  to each other, $|E(S)| = |C_{M,H}|$. Also $S \subseteq H$, and
  by Lemma~\ref{subpc}, $\text{cyc}(S) \geq |E(S)|/2$. Finally, from
  Lemma~\ref{alsub}, $\text{cyc}(T) \geq \text{cyc}(S) \geq |E(S)|/2 =
  |C_{M,H}|/2$.

  {\smallskip\noindent\em Step 2.} Now we remove all the edges in
  $E^{cr}_{M,H}$. Let $M_S$ be an arbitrary median of $S$, $M_T$ the
  matching for $T$ defined by the union of $M_S$ and an arbitrary
  completing matching for $T$ respective to $S$, and $M' = (M -Y(M))
  \cup M_S \cup M_T$.

  {\smallskip\noindent\em Claim 3. $M'$ is a median of $B$.}  First,
  by removing the edges in $Y(M)$, the total number of
  alternating cycles decreases by $|C_{M,H}|$. Next, $M_S$ and $M_T$
  contain at least $|C_{M,H}|/2$ alternating cycles each (Claim 2
  above). Hence, the new matching $M'$ contains at least the same
  number of alternating cycles than $M$. By definition of a median,
  $M'$ can not contain more alternating cycles than $M$, so it
  contains the same number of alternating cycles, and is a median of
  $B$. Note that this also implies that $\text{cyc}(S) = \text{cyc}(T)
  = \frac{|C_{M,H}|}{2}$.
  
  %{\smallskip\noindent \bf Ahmad. I am not sure of what you try to
   % prove here. Could you make it a claim as above.}
  {\smallskip\noindent\em Claim 4. $X(M') \subset X(M)$ and $X(M')
    \neq X(M)$.}  If there exists $e \in \text{X}(M') - \text{X}(M)$
  then there would be at least one $H$-crossing cycle induced by $M'$
  which is not induced by $M_S$ or $M_T$, this implies $B_{M'}$ would contain
  more alternating cycles than $B_{M}$, which contradicts the fact
  that $B_M$ and $B_{M'}$ have the same number of alternating cycles.
  Next, $\text{X}(M') \subset \text{X}(M)$, as $E(S) \subset
  \text{X}(M)$ and $E(S) \cap \text{X}(M') = \emptyset$ (the vertices
  in $S$ are matched to themselves). Therefore, $|X(M')| < |X(M)|$.

  By iterating the above steps we obtain a median with no crossing
  cycle. Then, by case 1, we can modify this median to a median
  without $H$-crossing edge. %% \qed
\end{proof}

%% \begin{remark}\label{proc}
%%   The transformation introduced in the proof of
%%   Proposition~\ref{indep} can be applied as long as there are at least
%%   two $H$-crossing edges. Indeed, if there is no $H$-crossing cycle
%%   and there are at least two $H$-crossing edges, we can remove two of
%%   them and match respectively their end-vertices in $H$ together and
%%   their remaining end-vertices together. Otherwise, if there is at
%%   least one $H$-crossing cycle we can define $S$ and $T$ as before and
%%   proceed with the transformation described in the proof.
%%   \end{remark}

\begin{proposition}\label{1dep}
  Let $H$ be a connected component of $B$. If $H$ is isomorphic to
  $P_{2k-1}$, for $k \geq 1$, then $H$ is 1-independent. 
\end{proposition}

\begin{proof}
  We follow the same proof strategy than for
  Proposition~\ref{indep}. The number of $H$-crossing edges is odd. If
  there is no $H$-crossing cycle, we can remove an even number of them
  as in case 1 of the proof of Proposition~\ref{indep}, leaving only
  one $H$-crossing edge. Otherwise, if we assume that there are
  $H$-crossing cycles, we can apply the transformation defined in case
  2 of the proof of Proposition~\ref{indep}. It has similar
  properties, as, from Lemma~\ref{subpc}, for every subgraph $H'
  \subseteq P_{2k-1}$, $\text{cyc}(H') \geq |E(H')|/2$, which implies
  again that $\text{cyc}(S) = \text{cyc}(T) = \frac{|C_{M,H}|}{2}$.
  %% \qed
\end{proof}

%The results above lead to the following proposition on the structure
%of a median graph.

\begin{proposition}\label{themedian}
  If $B$ contains only cycles and paths, there exists a median of $B$
  in which even components have no crossing edge, and each odd path
  has exactly one crossing edge.
\end{proposition}

\begin{proof}
  This result follows from applying, on an arbitrary median graph, the
  transformation introduced in the proof of in Proposition~\ref{indep}
  to each even/odd path or even cycle of the breakpoint graph,
  reducing then the number of crossing edges for each of them, without
  increasing the number of crossing edges in other components.  %% \qed
\end{proof}

%%%%%%%%%%%%%%%%%%%%%%%%%%%%%%%%%%%%%%%%%%%%%%%%%%%%%%%%%
%\subsection{Computing $\text{cyc}(C_{2k})$ and $\text{cyc}(P_k)$}
\subsection{Alternating cycles for arbitrary paths and even cycles}

The results of the previous section open the way to computing a median
of a breakpoint graph with maximum degree $2$ by considering each path
or even cycle independently, and matching odd paths into pairs (each
defined by a single crossing edge).  The main point of the current
section is to show that paths and even cycles are easy to consider
when computing a median.

\begin{proposition} \label{cycP}
  If $H\subseteq B$ is isomorphic to $P_k$, for some $k\geq 1$, then
  $\text{cyc}(H) = \lfloor \frac{k}{2}\rfloor$. Moreover, there exists
  a median whose edges in $H$ define $\lfloor \frac{k}{2}\rfloor$
  alternating 2-cycles, and one crossing edge incident to a terminal
  vertex of $H$ if $k$ is odd.
\end{proposition}

\begin{proof}
  From Lemma~\ref{subpc}, $\text{cyc}(H) \geq \lfloor
  \frac{k}{2}\rfloor$. We use induction on $k$ to show that
  $\text{cyc}(H) \leq \lfloor \frac{k}{2}\rfloor$. This obviously
  holds for $k = 1$. So we assume that $k \geq 2$, and consider a
  median $M$ for $H$. If there is no 2-cycle (an alternating cycle
  consisting of two parallel edges) in $H_M$, each alternating cycle
  has length at least $4$, and hence at least 2 colored edges. So
  $\text{cyc}(H) \leq \lfloor \frac{|E(H)|}{2} \rfloor = \lfloor
  \frac{k-1}{2} \rfloor \leq \lfloor \frac{k}{2}\rfloor$.

  Now assume that the median $M$ contains a 2-cycle, with vertices $u$
  and $v$. Shrinking $\{u, v\}$ results in $H'$ that is either a
  single path with $k-2$ vertices or two paths with $p$ and $q$
  vertices such that $p+q = k-2$.  In both cases, using induction and
  the fact that all paths are 0-independent or 1-independent, we can
  conclude that,
  \begin{itemize}
  \item if $H'$ contains one path, $\text{cyc}(H') \leq \lfloor
    \frac{k-2}{2} \rfloor + 1 = \lfloor \frac{k}{2}\rfloor$,
  \item if $H'$ contains two paths, $\text{cyc}(H') \leq \lfloor
    \frac{p}{2} \rfloor + \lfloor \frac{q}{2} \rfloor + 1 \leq \lfloor
    \frac{k}{2}\rfloor$.
  \end{itemize}
  To obtain a median with exactly $\lfloor \frac{k}{2}\rfloor$
  alternating cycles in $H$, we can simply define median edges by
  linking successive vertices in $H$ (as in the proof of
  Lemma~\ref{subpc}). If $k$ is odd this forces the unique
  $H$-crossing edge (Proposition~\ref{1dep}) to contain the last end
  vertex of $H$ (one of its two end vertices), which has no impact on
  the number of alternating cycles as, by definition, this crossing
  edge will not belong to any alternating cycle.  %% \qed
\end{proof}

%% \begin{remark}
%%   Note that if $B$ is isomorphic to $P_k$, then $\text{cyc}(B)$ is
%%   independent from the edge coloring of $B$, and then computing $m(G)$
%%   can be done in time $O(k)$.
%% \end{remark}

\begin{lemma} \label{cycC}
  If $B$ is isomorphic to $C_{2k}$, for some $k\geq 1$, then either
  $\text{cyc}(B) = {k}$ or $\text{cyc}(B) = {k} +1$.
\end{lemma}

\begin{proof}
  Obviously, $\text{cyc}(B) \geq {k}$. So, we assume that
  $\text{cyc}(B) > {k}$. Let $M$ be an arbitrary median of
  $B$. Following the proof of Proposition~\ref{cycP}, if all
  alternating cycles in $B_M$ have length at least 4, then the number
  of alternating cycles is at most ${k}$, so there must exist at least
  one 2-cycle in $B_M$.  Let $uv$ be a colored edge in a 2-cycle:
  $\text{cyc}(B) = \text{cyc}(B\cdot\{u, v\}) + 1$
  (Proposition~\ref{Shrink}).  Moreover $B\cdot\{u, v\}$ is a path, or
  a cycle, and it is a cycle if and only if the two edges incident to
  the ends of $uv$ have the same color. If it is a path,
  Proposition~\ref{cycP} implies that $\text{cyc}(B\cdot\{u, v\}) =
  {k-1}$ and $\text{cyc}(B) = {k}$, which contradicts the assumption
  that $\text{cyc}(B) > {k}$. So $B\cdot\{u, v\}$ is a cycle and the
  edges incident to $uv$ have same color. By induction on $k$ (note
  that $\text{cyc}(C_4)=3$ and $\text{cyc}(C_2)=2$) we can find a
  median of $B\cdot\{u, v\}$ with $\text{cyc}(B\cdot\{u, v\}) = {k-1}
  + 1=k$ or $\text{cyc}(B\cdot\{u, v\}) = {k-1}$, alternating
  cycles. Hence, $\text{cyc}(B) = {k} + 1$ or $\text{cyc}(B) =
  {k}$. %% \qed
\end{proof}

%% \begin{remark}
%%   By previous theorem, if $B$ is isomorphic to $C_{2k}$, then
%%   $\text{cyc}(B)$ dependents on the edge coloring of $B$.
%% \end{remark}

Some definitions below assume that cycles of $B$ are oriented, so we
assume from now that edges of every cycle of $B$ are consistently
oriented, clockwise or counterclockwise. Fig.~\ref{cyckind} provides
an illustration.

\begin{definition}\label{def:kind-cyc}
  A cycle $C_{2k}$ of $B$ is of the {\it first kind} if
  $\text{cyc}(C_{2k}) = {k}$, and it is of the {\it second kind} if
  $\text{cyc}(C_{2k}) = {k} + 1$.
\end{definition}

\begin{definition}\label{def:signature}
  Let $C$ be a cycle of $B$. The {\it signature} of a vertex of $C$ is
  an ordered pair $(a, b)$ such that $a$ and $b$ are the colors of the
  edges incident to that vertex: $a$ is the color of the incoming edge
  and $b$ the color of the outgoing edge.  Two vertices $u$ and $v$
  are {\it diagonal} if their signatures are of the form $(a, b)$ and
  $(b, a)$.
\end{definition}

\begin{definition}\label{def:cross} 
  Let $M$ be a median of an even cycle $C$, and $uv$ and $u'v'$ be
  edges in $M$: $uv$ and $u'v'$ {\it cross} if $u, u', v, v'$ appear
  in this order along $C$. A {\it cross-free diagonal} matching for
  $C$ is a matching whose edges connect pairs of diagonal vertices and
  no two edges cross.
\end{definition}

\begin{figure}{
  \begin{center}
    \includegraphics[scale=0.85]{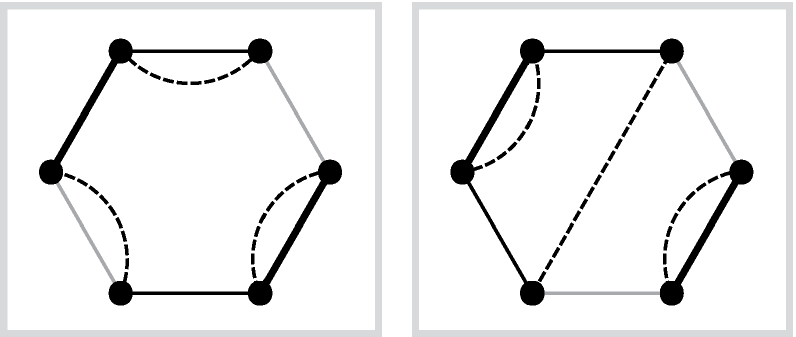}
  \caption{Dashed edges are median edges: (Left) a cycle $C_6$ of the
    first kind --- (Right) a cycle $C_6$ of the second kind; the
    matched vertices are diagonal}
       \label{cyckind}
  \end{center}
 }
\end{figure}
%%%%%%%%%

%% We first give an auxiliary lemma regarding cycles of the second kind.

\begin{lemma} \label{2ndKind}
  Let $B$ be isomorphic to an even cycle of the second kind, and $M$
  be a median of $B$. (1) Each edge in $M$ joins two diagonal
  vertices, and (2) the edges in $M$ do not cross.
\end{lemma}

\begin{proof}
  Let $B = C_{2k}$.  We first prove point (1) by contradiction. Assume
  that $uv \in M$ and that $u$ and $v$ are not diagonal. Let $(a, b)$
  and $(c, d)$ be the respective signatures of $u$ and $v$. Our
  assumption implies that $(c, d) \neq (b, a)$, and we can distinguish
  two cases: $(a, b) = (c, d)$ and $(a, b) \neq (c, d)$.
  \begin{itemize}
  \item If $(a, b) = (c, d)$, by shrinking the pair $\{u, v\}$ we
    obtain a smaller cycle $C_{2k-2}$, and by Proposition~\ref{Shrink},
    $\text{cyc}(B) = \text{cyc}(C_{2k}\cdot\{u, v\}) =
    \text{cyc}(C_{2k-2}) \leq \frac{2k-2}{2} + 1 = {k} $ which
    is a contradiction, since $B$ is of the second kind. Note that in
    this case $u$ and $v$ cannot be consecutive vertices on $B$.
  \item If $(a, b) \neq (c, d)$, by shrinking the pair $\{u, v\}$, the
    resulting graph can be either a path with $2k-2$ vertices, or a
    cycle and a path, together with $2k-2$ vertices.  In the first
    case, vertices $u$ and $v$ must be consecutive on $B$. But now
    $\text{cyc}(B) = \text{cyc}(C_{2k}\cdot\{u, v\}) + 1 =
    \text{cyc}(P_{2k-2}) + 1 = \frac{2k-2}{2} + 1 < {k} + 1$, which is
    a contradiction, since $B$ is of the second kind. In the second
    case $\text{cyc}(B) = \text{cyc}(C_{2k}\cdot\{u, v\}) =
    \text{cyc}(C_\ell) + \text{cyc}(P_m) \leq \frac{\ell}{2}+1 +
    \frac{m}{2} \leq \frac{2k-2}{2} + 1 < {k} + 1$, since paths are
    either 0- or 1-independent and $\ell+m=2k-2$ (note that in the
    latter case $u$ and $v$ cannot be consecutive). This is again a
    contradiction, as $B$ is of the second kind.
  \end{itemize}

  We now prove point (2). $B$ is of the second kind, as shown in the
  proof of Lemma~\ref{cycC}, there is a 2-cycle containing a colored
  edge $u'v'$. Moreover, by point (1), vertices $u'$ and $v'$ are
  diagonal. So $B\cdot\{u', v'\}$ is isomorphic to $C_{2k-2}$ and it
  must be of the second kind, as otherwise $\text{cyc}(C_{2k}) =
  \text{cyc}(C_{2k-2}) + 1 = \frac{2k-2}{2} + 1 = {k} < {k} +
  1$. Obviously, $u'v'$ does not cross with any median edge of $M$. By
  shrinking this pair and, by induction on the length of the cycle,
  applied to $C_{2k}\cdot\{u', v'\}$, the proof is complete. %% \qed
\end{proof}

%{\smallskip\noindent\bf I stopped here for today (Cedric, \today). We
 % should correct below to use $2k$ for even cycles instead of $k$.}

\begin{lemma} \label{2ndKindcyc}
  Let $B$ be isomorphic to  $C_{2k}$. $B$ is of the second kind if
  and only if there exists a matching $M$ of $B$ that is cross-free
  diagonal.
\end{lemma}

\begin{proof}
  The necessity follows from Lemma~\ref{2ndKind}.  Now assume that
  there exists a cross-free diagonal matching $M$ on vertices of $B$.
  It is easy to see that $M$ contains at least one edge $uv$ where $u$
  and $v$ are consecutive on $B$ (note that $M$ is a perfect matching,
  since $B$ has even number of vertices). If we shrink the pair $\{u,
  v\}$, the resulting graph is $C_{2k-2}$ and the remaining edges of
  $M$ are a cross-free diagonal matching for $C_{2k-2}$.  We can
  complete the proof by induction on $k$, since $\text{cyc}(C_{2k}) =
  1 + \text{cyc}(C_{2k-2}) = 1 + \frac{2k-2}{2} + 1 = k + 1$, and the
  statement of the lemma is obviously true for $k = 1$ and $k=2$.
  %% \qed
\end{proof}

\begin{lemma} \label{WhatKind}
  Let $B$ be isomorphic to $C_{2k}$. Deciding if $B$ admits a
  cross-free diagonal matching can be done in time $O(k)$.
\end{lemma}

\begin{proof}
  Let $B = v_1v_2\ldots v_{2k}$. We rely on a simple greedy algorithm,
  which is in fact a classical algorithm for deciding if a circular
  parenthesis word is balanced; we present it for the sake of
  completeness.

  The key point was given in the proof of Lemma~\ref{2ndKindcyc}: any
  cross-free diagonal matching contains at least one pair of
  consecutive vertices that are matched. Given the circular nature of
  $B$, we can extend this property as follows: if $u$ and $v$ are
  consecutive diagonal vertices and $B$ admits a diagonal cross-free
  matching, then there exists a matching where $u$ and $v$ are
  matched. This leads immediately to a greedy algorithm that matches
  such vertices as soon as they are visited, using a simple stack data
  structure:
 
  %% \begin{center}
  %%   \scalebox{1}{ \fbox{\begin{minipage}{0.98\linewidth}
  %%         \begin{enumerate}
  %%         \item $M=\emptyset$ ($v_1,v_2,\ldots,v_{2k}$ are not matched)
  %%         \item For $j = 1$ to $2k$
  %%           \begin{enumerate}
  %%           \item if there exists $i$ ($1 \leq i < j$), such that
  %%             $v_i$ and $v_j$ are diagonal, $v_i$ is not matched, and
  %%             $i$ is the maximum number with this property, then add
  %%             $\{v_i,v_j\}$ to $M$ (match $v_i$ and $v_j$).
  %%           \end{enumerate}
  %%         \item If all vertices are matched, $B$ has a cross-free
  %%           diagonal matching, otherwise it does not.
  %%         \end{enumerate}
  %%       \end{minipage}
  %%   } }
  %% \end{center}
  
  %% The time complexity of this algorithm is $O(k^2)$: we iterate the
  %% loop $2k$ times and for each $j$ in the loop we check previous
  %% vertices to find the proper $i$. One can easily see that this can be
  %% done in linear time as follows:
 
  \begin{center}
    \scalebox{1}{ \fbox{\begin{minipage}{0.9\linewidth}
          \begin{enumerate}
          \item Let $M=\emptyset$ be an empty matching.
          \item Let $S$ be an empty stack.
          \item For $j = 1$ to $2k$
            \begin{enumerate}
            \item if the top element $v_i$ of $S$ is diagonal with
              $v_j$, pop it from the stack $S$ and add $\{v_i,v_j\}$
              to $M$.
            \item else, push $v_j$ on $S$.
            \end{enumerate}
          \item If $S$ is empty, $B$ admits a cross-free diagonal
            matching, given by $M$, otherwise it does not admit one.
          \end{enumerate}
        \end{minipage}
    } }
  \end{center}

  The time complexity of this algorithm is obviously linear in $k$.
  %% \qed
\end{proof}

\begin{proposition}\label{cycles}
  If $B$ is isomorphic to an even cycle of size $k$ ($k\geq 2$), then
  computing $\text{cyc}(C)$ can be done in time $O(k)$.
\end{proposition} 

\begin{proof}
  Immediate consequence of Lemma~\ref{cycC}, Lemma~\ref{2ndKindcyc},
  and Lemma~\ref{WhatKind}.%% \qed
\end{proof}

%%%%%%%%%%%%%%%%%%%%%%%%%%%%%%%%%%%%%%%%%%%%%%%%%
\subsection{Proof of Theorem~\ref{thm:main}}

We now have all the elements to prove our main result,
Theorem~\ref{thm:main}. We first prove that computing a median of a
breakpoint graph of maximum degree two is tractable.

\begin{lemma}\label{linkage}
  If $B$ has maximum degree $2$, then there exists a median of $B$ such
  that every odd connected component of $B$ is connected by median
  edges to exactly one other odd connected component.
\end{lemma}

\begin{proof}
  Let $M$ be a median as described in the proof of
  Proposition~\ref{themedian}: every even connected component has no
  crossing edge and each odd path has exactly one crossing
  edge. Moreover, odd cycles have at least one crossing edge.

  Let $H$ be an odd connected component and $e$ one of its crossing
  edges, connecting $H$ to another odd component $H'$. Shrinking $e$
  results into $(H \cup H')\cdot e$ which is a set of even components
  and it is then $0$-independent. Moreover, as $H$ and $H'$ were
  distinct connected components of $B$, from Proposition~\ref{Shrink}
  (with $k=0$), $\text{cyc}((H \cup H')\cdot e)=\text{cyc}(H \cup
  H')$. 

  Repeating this argument for other odd components and the fact that
  the number of odd components is even (because the number of vertices
  in the breakpoint graph is even) completes the proof. %% \qed
\end{proof}

\begin{lemma}\label{pair}
  If $B$ has maximum degree $2$ and consists of two odd connected
  components $H_1$ and $H_2$, of respective sizes $k_1$ and $k_2$,
  then computing a median of $B$ can be done in time $O(k_1k_2(k_1+k_2))$.
\end{lemma}

\begin{proof}
    For parity reasons, a median $M$ contains at least one edge $e$
    between $H_1$ and $H_2$ ($e$ is a $H_1$-crossing edge).%%  If one of
    %% $H_1$ and $H_2$ is a path, then, from Proposition~\ref{cycP}, we
    %% can assume $e$ is the only crossing edge of $M$ is connected to
    %% one of its two terminal vertices.  In either case, by 
    By shrinking $e$ we obtain either one even connected component or
    two even connected components, and, from
    Proposition~\ref{themedian}, we can compute a median for each
    connected component independently. This computation requires
    linear time (Propositions~\ref{cycP} and~\ref{cycles}). There are
    at most $k_1k_2$ possible candidates for
    $e$. Hence computing a median of $B$ is tractable in time
    $O(k_1k_2(k_1+k_2))$. %% \qed
\end{proof}

%% \begin{figure}[H]{
%% \begin{center}
%%  \includegraphics[scale=0.25]{comps.pdf}
%%  \caption{{\bf (e)} Each even component has no crossing edge, {\bf
%%      (o)} For each odd component $H_1$, there is exactly one other odd
%%    component $H_2$, such that $(H_1 \cup H_2)$ has no crossing edge.}
%%  \label{comps}
%%  \end{center}
%%  }\end{figure}

\begin{proposition}\label{odds}
  If $B$ is a breakpoint graph with $2n$ vertices with maximum degree
  $2$, then computing a median of $B$ can be done in $O(n^3)$.
\end{proposition}

\begin{proof}
    We first consider the case where $B$ contains only odd connected
    components. We define a complete edge-weighted graph $K_B$ as
    follows:
   \begin{enumerate}
    \item each connected component $C$ defines a vertex $v_C$;
    \item each edge $\{v_C,v_D\}$ has weight  $\text{cyc}(C \cup D)$
    \end{enumerate}
    By Lemma~\ref{pair}, $K_B$ is computable in polynomial time. We
    claim it is computable in $O(n^3)$.  Suppose $B$ has $t$
    components and $n_1, \ldots, n_t$ are the number of vertices in
    each component. So we have $n_1 + \ldots +n_t = 2n$. The time to
    construct $K_B$ is of order
    
    \begin{align*}
      \sum\limits_{i < j} n_i\cdot n_j\cdot (n_i + n_j)= \sum\limits_{i < j} n_i^2\cdot n_j + n_i\cdot n_j^2 \\ = 
      \frac{1}{3}((2n)^3 - (n_1^3 + \ldots + n_t^3))
      % \leq \frac{1}{2}(4n^2 - t\cdot (2n/t)^2) 
      \leq \frac{8}{3}n^3.
    \end{align*}
    
    Finally, by Lemma~\ref{linkage} we only need to find a maximum
    weight matching for $K_B$, which can be done in $O(n^{3})$ by
    using Edmonds's algorithm~\cite{Edmonds}.  %% \qed

    If the breakpoint graph $B$ has maximum degree $2$, its connected
    components are paths or cycles. From Proposition~\ref{themedian}
    and Proposition \ref{cycles} we can find the median edges for even
    components independently. Finally for odd components we find the
    median edges as described in the first part of the proof. %% \qed
\end{proof}

\noindent{\em Proof of Theorem~\ref{thm:main}.}  We now assume that
$B$ has maximum degree $3$.  

The main idea is to consider all possibilities for matching the
vertices of degree $3$ of $B$. A vertex $u$ of degree $3$ can be matched
in two ways.
\begin{itemize}
\item If it is matched to another vertex of degree $3$, by shrinking
  these two vertices we obtain a smaller graph with fewer vertices of
  degree $3$, and, from Proposition~\ref{Shrink}, we know precisely
  the number of alternating cycles (here $2$-cycles) lost in the
  shrinking process, given by the number of genome edges between the
  two shrinked vertices.
\item If it is matched to a vertex of degree less than 3, then one of
  the edges incident to $u$ is not in any alternating cycle, and we
  can remove this edge and transform $u$ into a vertex of degree $2$
  (Fig.~\ref{3edge}).
\end{itemize}

\begin{figure} 
  \begin{center}
    \includegraphics[scale=.35]{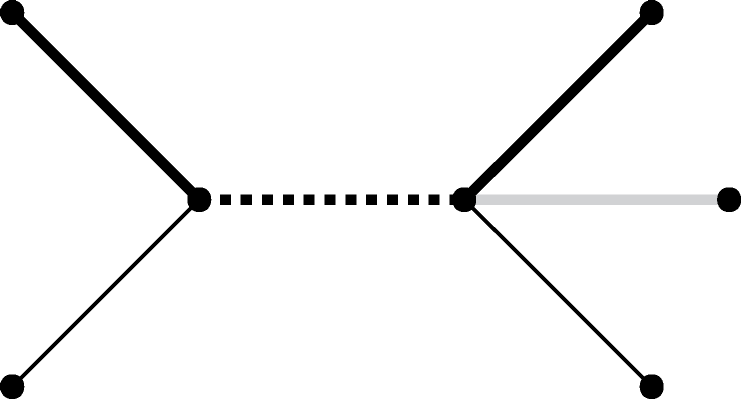}
    \caption{The dashed edge is a median edge. The gray edge cannot be
      in any alternating cycle.}
    \label{3edge}
  \end{center}
\end{figure}

Now for each $i$, $0 \leq i \leq \ell$, we can select $2i$ vertices
among all $m$ vertices of degree $3$ (there are $O(m^{2i})$
possibilities), compute an arbitrary perfect matching on these $2i$
vertices, and, for each each remaining vertex of degree $3$, remove an
edge incident to this vertex (there are $O(3^{m-{2i}})$
possibilities). The resulting breakpoint graph $B'$ is of maximum
degree $2$ and a median can be computed in time $O(n^3)$, whose number
of alternating cycles needs only to be augmented by the number of
edges between matched vertices of degree $3$ in $B$. 

The number of all such configurations is in
$O((\ell+1) \cdot (m^{2\ell}+1) \cdot 3^m)$ (the term $+1$ is needed to
account for the case $m=0$), which leads to the stated complexity.
\section{Conclusion}

In this work, we characterized a large class of tractable instances
for the DCJ median problem (with circular median and mixed
genomes). In fact, we showed that only the vertices of degree $3$ make
the problem intractable. Also, by removing $k$ edges from the
breakpoint graph and decreasing its maximum degree, cost of its median
is not bigger than $k$ plus the cost of the main median (i.e. the
current cost $ - k$ is a lower bound for the cost of the main median).
Finally, we showed there is an FPT algorithm for the DCJ median
problem, if there exists a median such that the number of its edges
connecting two vertices of degree $3$ is bounded.

Our work also shows that the multiplicity of solutions (i.e. medians)
is likely to happen when dealing with breakpoint graphs with long
paths or even cycles, as we showed that such components can admit
several optimal medians. Hence, our results, as they stand now, are of
interest more for computing the score of a median than for computing
actual medians that can be seen as realistic ancestral
genomes. However, the problem of uniform sampling of optimal median is
worth being explored, even in the simpler setting of breakpoint graphs
of maximum degree $2$ in a first time.

From a theoretical point of view, our work raises several
questions. First, it leaves open the possibility that the DCJ median
problem is FPT. Using the number of vertices of degree $3$ as a
parameter is a a natural approach, although this seems to be a
difficult question to address. The next obvious problem is to extend
our approach to the case of a mixed or linear median. This would
require to better understand the combinatorics of odd paths in the
breakpoint graphs in relation to medians. The simpler problem to find
an optimal way to remove exactly one edge from each circular
chromosome of a circular median while minimizing the number of
destroyed alternating cycles is also open. Extending our results to
the related \emph{DCJ halving problem}~\cite{Tannier2009} is also a
natural question.

Another interesting question is about expanding the breakpoint
distance toward the DCJ distance: for two genomes $G_1$ and $G_2$ on
$n$ genes, their breakpoint distance is equal to
$$d_{\text{BP}}(G_1, G_2) = n - a(G_1, G_2) - \frac{1}{2}e(G_1,
G_2).$$ The parameters $a(G_1, G_2)$ and $e(G_1, G_2)$ are also equal
the number of 2-cycles and 1-paths ($P_1$) in the breakpoint graph
$B(G_1, G_2)$, respectively.  The DCJ distance of these genomes is:
$$d_{\text{DCJ}}(G_1, G_2) = n - c(G_1, G_2) - \frac{p(G_1,
  G_2)}{2},$$ where $c(G_1, G_2)$ and $p(G_1, G_2)$ are the number of
(even) cycles and odd paths in the $B(G_1, G_2)$, respectively. This
motivates us to define a dissimilarity function as follows:
$$d_{(i, j)}(G_1, G_2) = n - c_i(G_1, G_2) - \frac{1}{2}p_j(G_1,
G_2),$$ where $c_i(G_1, G_2)$ is the number of (even) cycles with at
most $2i$ vertices, and $p_j(G_1, G_2)$ is the number of odd paths
with at most $2j-1$ vertices.  By considering this dissimilarity
measure, the median problem is tractable when $i = j = 1$, since
$d_{(1, 1)} = d_{\text{BP}}$. By taking $i = j = \infty$ we have
$d_{(\infty, \infty)} = d_{\text{DCJ}}$, and the median problem would
be intractable. A natural question is then to understand for which
values of $i$ and/or $j$ the median problem is tractable, or FPT.

\bigskip{\noindent\bf Acknowledgments.}
C.C. and L.S. are supported by NSERC Discovery Grants.

%%%%%%%%%%%%%%%%%%%%%%%%%%%%%%%%%%%%%%%%%%%%%%%%%%%%%%%%%%

\end{document}